\documentclass[11pt,a4paper]{article}

\usepackage[verbose,tmargin=1in,bmargin=1in,lmargin=1in,rmargin=1in]{geometry}
\usepackage[dvipsnames,table]{xcolor}
\usepackage[colorlinks=true,linkcolor=blue!50!black,citecolor=blue!50!black,urlcolor=blue!50!black]{hyperref}

\usepackage{amsmath,amsthm,amsfonts,amssymb,stmaryrd}
\usepackage{cancel,mathrsfs,calligra}
\usepackage{mathtools}
\usepackage{slashed}
\usepackage{enumitem}

\newcommand{\N}{\mathbb{N}}
\newcommand{\R}{\mathbb{R}}
\newcommand{\C}{\mathbb{C}}
\newcommand{\Z}{\mathbb{Z}}

\newcommand{\sk}[1]{\left\langle{#1}\right\rangle}

\newcommand{\Tr}{\operatorname{Tr}}
\newcommand{\tr}{\operatorname{tr}}
\newcommand{\im}{\operatorname{Im}}
\newcommand{\id}{\operatorname{id}}

\newcommand{\cH}{\mathcal{H}}
\newcommand{\cC}{\mathcal{C}}

\newtheorem{fact}{Fact}[section]

\newtheorem{lemma}[fact]{Lemma}

\newcommand{\Aref}{{A_{\operatorname{ref}}}}

\newcommand{\re}{\operatorname{Re}}

\newcommand{\Res}{\operatorname{Res}}

\usepackage[colorinlistoftodos,textsize=scriptsize]{todonotes}

\newcommand{\sA}{S_A}
\newcommand{\sB}{S_\Aref}
\newcommand{\sAf}{S_{A+F}}
\newcommand{\sAg}{S_{A+G}}
\newcommand{\sAfg}{S_{A+F+G}}
\newcommand{\inv}{{-1}}
\newcommand{\Pmm}{{-}}

\newcommand{\Ppm}{{+-}}
\newcommand{\Pmp}{{-+}}

\begin{document}

\title{\textsc{Vacuum polarisation without infinities}}
\author{Dirk - André Deckert, Franz Merkl, Markus Nöth\thanks{Department of Mathematics, LMU Munich}}

\maketitle

\begin{abstract}
    In honour of Detlef Dürr, we report on a mathematical rigorous computation of the electric vacuum polarisation
    current and extract the well-known expression for the second order perturbation. Intermediate steps in the presented
    calculation demonstrate, to the knowledge of the authors for the first time, mathematical rigorous versions of the
    combined dimensional and Pauli-Villars regularisation schemes. These are employed as computational tools to infer
    convenient integral representations during the computation. The said second order expression is determined up to a
    remaining degree of freedom of a real number -- without ill-defined terms from start to end.
\end{abstract}

\section{Introduction}

The definition and the original computations of the electric vacuum polarisation current in quantum electrodynamics
(QED), based on the pioneering works of Dirac, Heisenberg, and others, go way back to Schwinger, Feynman, and Dyson,
and, in its original form, may still be best accessible from Dyson's manuscript \cite{Dyson2006}. Today, these are
contained in various nuances in almost every textbook on advanced quantum mechanics. It might therefore appear that this
topic has long been settled. Given that QED is such an important theory for the
human understanding of nature and a
century has passed, it better should have. But, at least mathematically, it has not. All of these computations start with
an ill-defined equation of motion for the electric quantum vacuum, or worse, an ill-defined and far fetched scattering
matrix coefficient, derive a similarly ill-defined expression for the electric current which is then massaged by means
of several informal manipulations, such as ``subtracting'' an ill-defined zero-field electric current, ``dropping''
ill-defined expressions that do not appear gauge invariant, ``introducing'' diverging counter terms to absorb remaining
infinities into the bare electric charge constant, etc. For everyone who has had the opportunity to get to know Detlef, it
probably goes without saying that he was very unsatisfied with this state of
affairs. At the same time, of course, he could not have
cared less whether the consensus opinion of the scientific community was otherwise. In fact, beside his renown work in
the foundations of quantum mechanics, he dedicated a substantial part of his scientific work to the understanding of
classical and quantum electrodynamics. In 2006, Detlef, Martin Schottenloher, D.-A.D.\ and F.M.\ started a seminar
series on the mathematics of QED based on the books \cite{Dyson2006,Thaller1992,Scharf1995,loop1986}, articles 
\cite{Dirac1934,Shale1965,Ruijsenaars1977,Fierz1979,scharf1986causal,Wurzbacher2001}, among others, 
entirely in the spirit of Dirac's quote \cite[p. 184]{kragh1990dirac}:
\begin{quote}
    ``I must say that I am very dissatisfied with the situation, because this so-called \emph{good theory} does involve
    neglecting infinities which appear in its equations, neglecting them in an arbitrary way. This is just not sensible
    mathematics. Sensible mathematics involves neglecting a quantity when it is small – not neglecting it just because
    it is infinitely great and you do not want it!'' \hfill \emph{-- Dirac, 1975,}
\end{quote}
which initiated the works \cite{Deckert2010a,Deckert2014,deckert_perspective_2015,deckert2016external}. Many years later,
we are now under the impression of having a mathematical rigorous as well as non-perturbative understanding of this
computation as well as the definition of the corresponding time evolution and scattering matrix for QED in an external,
classical electrodynamic field, which will be published in a series of forthcoming articles. This first article shall
provide an introduction by treating only the second order of perturbation of the electric vacuum polarisation current in
our approach. A slightly different but also rigorous computation of the second order of perturbation was already done by
Scharf \cite{Scharf1995} while the non-perturbative computation in \cite{Scharf1986} seems incomplete.

Along the way we have learned about many other series of works on various aspects of the mathematical rigorous
description of the quantum vacuum, three of which we would like to mention here: First, the works of Mickelsson and
collaborators \cite{Langmann1996,Mickelsson1998,Mickelsson2001, mickelsson_phase_2014} developing a bundle theoretic,
conceptual geometric theory of the phases in quantum field theory. Second, the works \cite{nenciu1987} and
\cite{Pickl2008, pickl2008adiabatic} on adiabatic electron-positron pair-creation in an external field. Third, the long
series of articles by Gravejat, Hainzl, Lewin, Séré, and Solovej which mainly studies the stationary solutions of a
non-linear model of the quantum vacuum, among them \cite{hainzl2004vacuum, hainzl2005self, hainzl2007minimization,
gravejat2012two, gravejat2013construction} and, in particular, the overview in \cite{lewin2014nonlinear}. Those models
are not only able to describe the polarisation of the quantum vacuum by an external field but also the back reaction of
its quantum expectation and make contact to effective dynamics in terms of the Heisenberg-Euler Lagrangian
\cite{gravejat2018derivation}. Third, an entirely different approach was developed by Finster and his group which
constructs the interacting fermionic projectors by means of a variational principle
\cite{finster2006principle,finster2010causal,finster2015non,finster2016non,Finster2015a,FinsterBook} and from which it
was shown that many aspects of quantum field theory, also beyond QED, can be derived rigorously. Lastly, viewing  QED in
external fields from the perspective of algebraic quantum field theory in curved space-time \cite{bar2007wave,
bar2009quantum}, the study of Hadamard states, see, e.g., \cite{schlemmer2015current}, can be seen as a bridge between the algebraic formulation, causal fermion
systems, and the approach followed here.

\section{Bogoliubov's electric current formula}

Before choosing a second-quantised expression of the electric charge current, for whatever that entails physically, it
is illustrative to discuss corresponding expressions in the one-particle quantum theory. For this purpose, we regard a
one-particle Schrödinger evolution of the form
\begin{align}\label{eq_schroedinger}
    i\partial_t \psi(t) =  H_A(t)\psi(t)
\end{align}
for times $t\in\R$, wave functions $t\mapsto\psi(t)$ with values in a Hilbert space $\cH$, and an external, classical,
time-dependent four-vector potential $A=(A_\mu)_{\mu=0,1,2,3}=(A_0,-\mathbf A)\in\cC^\infty_c(\R^4,\R^4)$. In the whole paper, we impose natural units
$1=\hbar=c=\mu_0=\epsilon_0$ which imply that charges are dimensionless, masses, momenta, energies and four-vector
potentials have the dimension $1/\text{length}$ while the Fourier transformed four-vector potentials carry dimension
$(\text{length})^3$. Furthermore, $(H_A(t))_{t\in\R}$ shall denote a family of Hamiltonians, i.e., possibly unbounded,
self-adjoint operators with domain $D(H_A(t))\subseteq \cH$. In the physically relevant cases regarded below, the latter
domain will be independent of time and coincide with the domain of the time-independent free Hamiltonian
$H_0=H_{A}|_{A=0}$.
Suppose, \eqref{eq_schroedinger} generates a strongly continuous unitary time-evolution $U_A=(U_A(t_1,t_0))_{t_0,t_1\in\R}$
on $\cH$, such that for a given initial value $\psi(t_0)$ in a suitable domain at time $t_0$, the corresponding unique
solution to \eqref{eq_schroedinger} is given by the map $t\mapsto \psi(t):=U_A(t,t_0)\psi$, one may introduce the
scattering operator 
\begin{align}
    S_A = 
    U_0(0,t_1) U_A(t_1,t_0) U_0(t_0,0) \,,
\end{align}
where $t_0\ll 0$ and $t_1\gg 0$ are taken so small and large, respectively,
that the temporal support of $A$ is contained in the interval $(t_0,t_1)$, and  
$U_0=U_A|_{A=0}$ is short for the corresponding free time-evolution. 

Using the notation $\partial_F g(F)=\frac{d}{d\epsilon}g(\epsilon F)|_{\epsilon=0}$ and viewing the charge current as
being generated in response to a perturbation of the four-vector potential $A$, one may define the evaluation of the
electric current at a test function $F=(F_0,-\mathbf F)\in\cC^\infty_c(\R^4,\R^4)$ as the operator expression
\begin{align}\label{eq_one-particle_current}
    J_A(F) := i e S_A^{-1} \partial_F S_{A+F}  \,,
\end{align}
which we refer to as Bogoliubov's formula of the electric current. On a suitable domain, this derivative can be
evaluated by employing the comparison of the dynamics of $A+F$ and $A$, i.e.,
\begin{align}
    \partial_F U_{A+F}(t_1,t_0) 
    &=
    -i\int_{t_0}^{t_1} d s \, U_A(t_1,s) \partial_F H_{A+F}(s) U_{A}(s,t_0) \, ,
    \label{eq_comparison_dyn}
\end{align}
for times $t_0,t_1\in\R$, and for which the appearance of the unbounded operators on the right-hand side will turn out
unproblematically in the physical relevant cases regarded below. Provided the family of Hamiltonians $(H_A(t))_{t\in\R}$ are
sufficiently regular in $A$, formulas in \eqref{eq_one-particle_current} and \eqref{eq_comparison_dyn} imply for $t\ll 0$ earlier than the temporal support of $A$
\begin{align}
    J_A(F)
    = \int\limits_{-\infty}^{+\infty} d s\,
        U_0(0,t)U_A(t,s) \partial_F H_{A+F}(s) U_A(s,t) U_0(t,0) \,.
    \label{eq_one_part_j_expression}
\end{align}
Below we state the well-known evaluations of \eqref{eq_one_part_j_expression}
for a minimally-coupled, non-relativistic, charged
spin-0 Schrödinger particle 
and a relativistic Dirac particle both of electric charge $-e<0$ and
in an external four-vector potential $A$:
\begin{enumerate}
    \item Schrödinger case, i.e., for $H_A(t)=H_A^S(t):=\frac{1}{2 m} (-i\nabla +
        e\mathbf A(t))^2-eA^0(t)$: 
        \begin{align}
            &j_A(F):=\langle \psi, J_A(F) \psi\rangle  
            = 
            \int\limits_{\R\times \R^3} dt\,d\mathbf x\, 
            \left[
                \rho(t,\mathbf x) F_0(t,\mathbf x)
                -\mathbf j(t,\mathbf x) \cdot \mathbf F(t,\mathbf x)
                \right],\nonumber
            \\
            &\rho(t,\mathbf x)
            =-e |\psi(t,\mathbf x)|^2, 
            \quad \text{and} \quad 
            \mathbf j(t,\mathbf x)
            =
            -\frac em
              \im \psi(t,\mathbf x)^*(\nabla+i e\mathbf A(t,\mathbf x))\psi(t,\mathbf x) 
            \, .
        \end{align}
    \item Dirac case, i.e., for $H_A(t)=H_A^D(t):=\boldsymbol\alpha \cdot
        (-i\nabla + e \mathbf A(t)) - e A^0(t) + \beta m$:
        \begin{align}
            j_A(F):=\langle \psi, J_A(F) \psi\rangle  
            = 
            \int\limits_{\R^4} j^\mu(x) F_\mu(x) d^4x,
            \qquad j^\mu(x) 
            = 
            -e \overline{\psi}(x) \gamma^\mu \psi(x)\,.
        \end{align}
\end{enumerate}

In conclusion, the electric currents of respective one-particle theories can be recovered from Bogoliubov's single
current formula \eqref{eq_one-particle_current}. By virtue of its generality, we will employ it as starting point to
infer a second-quantised version.

\section{Electric current of a Dirac sea in an external field}

In what follows, $U$ shall denote a unitary operator on one-particle Hilbert space $\cH$. Assuming some familiarity of the reader with the well-known second-quantisation
formalism of the Dirac field, e.g.~\cite{Thaller1992}, we avoid a lengthy introduction and only caricature the
respective Fock space vacuum expectation values. Let $P_+,P_-$
denote the orthogonal projections of the
one-particle Hilbert space onto the positive and negative spectral
subspaces $\cH^+$, $\cH^-$ of the free Dirac Hamiltonian $H^D_{0}$, respectively. The corresponding splitting of $\cH$ is denoted by $
\cH=L^2(\R^3,\C^4)=\cH^+\oplus\cH^-$ and the vacuum vector $\Omega$
in standard Fock space is taken to represent the Dirac
sea of $\cH^-$.
As a first guess, one might try to define the second-quantised vacuum
expectation value $\langle\Omega, \widetilde U\Omega\rangle$ of a second-quantised version $\widetilde U$ of $U$
by means of an infinite-dimensional Slater determinant, i.e., $\langle\Omega,
\widetilde U\Omega\rangle \stackrel{?}{=} \det U_\Pmm \big|_{\cH^-\to\cH^-}$
for $U_\Pmm:=P_-UP_-$. Interpreting the determinant as a
Fredholm determinant and using the notation $I_1(\cH)$ and $I_2(\cH)$ to denote the trace class and Hilbert-Schmidt
ideals, respectively,
this would require having $U_\Pmm\in \id_\Pmm + I_1(\cH^-)$. But this does not hold in general,
hence, the overset `?'. A way to fix this is to utilise another unitary operator $R_U$ on $\cH^-$, in our
notation $R_U\in U(\cH^-)$, interpreted as a base change in the vacuum, to define
\begin{align}
    \label{eq_vacuum_exp}
    \sk{\Omega, \widetilde U\Omega} := \det U_\Pmm  R_U
    \big|_{\cH^-\to\cH^-} \, .
\end{align}
This can indeed be shown to be the vacuum expectation value of some unitary lift $\widetilde U$ of the unitary
one-particle operator $U$ on $\cH$ to some Fock space; cf.~\cite{Deckert2010a}. In the special case of
\begin{align}
    \label{eq_ures_invert}
    U_\Ppm:=P_+UP_-,\, U_\Pmp:=P_-UP_+ \, \in I_2(\cH),
    \qquad
    \|\id-U\|<1 \Rightarrow {U^*}_\Pmm  \in 
    \operatorname{GL}(\cH^-) \, ,
\end{align}
it is possible to construct such an $R$ along the following lines. First, we
observe that
\begin{align}
    U_\Pmm  {U^*}_\Pmm  
    = {(U U^*)}_\Pmm  - U_\Pmp  {U^*}_\Ppm  
    = \id_\Pmm - |U_\Pmp  |^2
    \in 
    \id_\Pmm + I_1(\cH^-)
    \label{eq_one_plus_traceclass}
\end{align}
since the product of two Hilbert-Schmidt operators is in the trace class. Furthermore, the
expression~\eqref{eq_one_plus_traceclass} is positive definite thanks to~\eqref{eq_ures_invert}. However, ${U^*}_\Pmm$
is in general not unitary. Therefore, we employ a polar decomposition ${U^*}_\Pmm =R_U|{U^*}_\Pmm |$ with radial part
$|{U^*}_\Pmm |=\sqrt{{U}_\Pmm {U^*}_\Pmm } \in \id_\Pmm+I_1(\cH^-)$ having a positive Fredholm determinant. Exploiting
the invertibility of ${U^*}_\Pmm $, 
\begin{align}
    R_U={U^*}_\Pmm |{U^*}_\Pmm |^{-1}\in U(\cH^-)
    \label{eq_right_op}
\end{align}
renders the right-hand side of~\eqref{eq_vacuum_exp} well-defined. Furthermore, since $R_U\in U(\cH^-)$, it
turns out that the corresponding lift will also be unitary on the underlying Fock space; cf.~\cite{Deckert2010a}.

The lift $\widetilde U$ of $U$ is known to be unique up to a phase $e^{i\varphi}\in U(1)$;
cf.~\cite{Shale1965,Deckert2010a}. Hence, given $U\in U(\cH)$ that fulfils~\eqref{eq_ures_invert}, the pair $(U,R_U)$
characterises the equivalence class of lifts $[U,R_U]=\{(U,R_UQ)\,|\, Q\in U(\cH^-) \cap
\left(\id_{\cH^-}+I_1(\cH^-)\right)\}$ while two pairs $(U,R_U Q)$ and $(U,R_U Q')$
correspond to the same lift if and only if $\det {Q}^{-1}Q'=1$. The lift of $U$ characterised by $(U,R_U)$, i.e.,
$Q=\id_{\cH^-}$, shall be denoted by $\overline U$. This implies $\sk{\Omega,\overline U,\Omega}>0$
for~\eqref{eq_vacuum_exp}. For two $U,U'\in U(\cH)$ so close to $\id_{\cH}$ such that all three operators $U$, $U'$ and
$UU'$ fulfil~\eqref{eq_ures_invert}, and two pairs $(U,R)$ and $(U',R')$, characterising lifts $\widetilde{U}$ and
$\widetilde{U}'$, respectively, $(UU',R'R)$ characterises a lift of $UU'$ since ${(UU')}_\Pmm R'R= U_\Pmm U'_\Pmm
R'R+U_\Pmp U'_\Ppm R'R\in\id_{\cH^-}+I_1(\cH^-)$. Note the reversed order in the second component $R'R$. 
Moreover, $(U^{-1},R^{-1})$ characterises the lift ${\widetilde U}^{-1}$ of $U^{-1}$. We refer the interested reader to
\cite{Deckert2010a} and also to \cite{loop1986} for the underlying mathematical theory. 
We fix a reference vector potential $\Aref\in\cC^\infty_c(\R^4,\R^4)$ for the
rest of the article. We shall only
regard vector potentials $A\in\cC^\infty_c(\R^4,\R^4)$ such that the corresponding one-particle scattering
operators fulfil $\|S_\Aref^{-1}S_A-\id_{\cH}\|<1$, as a global construction
will not be needed here.
Moreover,
it is well-known, e.g., \cite{Ruijsenaars1977,Deckert2010a}, that ${S_A}_\Ppm,{S_A}_\Pmp\in
I_2(\cH)$ hold true for any $A\in\cC^\infty_c(\R^4,\R^4)$. In view of \eqref{eq_ures_invert},  this ensures
well-definedness of the lift
\begin{align}
    \widetilde S^A_{A+F} 
    = 
    \overline{S_\Aref^{-1}S_A}^{-1}\overline{S_\Aref^{-1}S_{A+F}}
    \label{eq_lift}
\end{align}
and allows a first attempt in defining a vacuum expectation of the current in the spirit of~\eqref{eq_one-particle_current}:
\begin{align}
    \widetilde{j}_A(F)=
    \sk{\Omega, \widetilde J_A(F)\Omega} 
    &=
    i\partial_F \sk{\Omega, \widetilde S^A_{A+F} \Omega}
    =\re i\partial_F \sk{\Omega, \widetilde S^A_{A+F} \Omega}
    \,.
    \label{eq_current}
\end{align}
We note that since $\widetilde S^A_{A+F}$ is a unitary lift on the
underlying Fock space which depends smoothly on $F$ in the relevant norm;
cf.~\cite{Deckert2010a}, the expression in the centre of \eqref{eq_current} is
real-valued.

However, taking different lifts $e^{i\varphi_A}\overline{S_\Aref^{-1}S_A}$, for an arbitrary, $A$-dependent
$\varphi_A\in\R$, gives
\[
    \widehat S^A_{A+F} 
    =
    {\left(e^{i\varphi_A}\overline{S_\Aref^{-1}S_A}\right)}^{-1}
    \left(e^{i\varphi_{A+F}}\overline{S_\Aref^{-1}S_{A+F}}\right)
\]
and yet another corresponding current
\begin{align}
    \widehat j_A(F)
    =
    i\partial_F \sk{\Omega, {\widehat S}^A_{A+F} \Omega} 
    =
    \widetilde j_A(F) - d \varphi_A(F) \, .
\end{align}
It is therefore the task to select physically relevant candidates for the physical current $j$ among those $\widehat j$
for the various phases $\varphi$. Morally, this non-uniqueness in the choice of the current 
reflects the ill-definedness of the current in the traditional formulation of QED. 

In order to characterise this degree of freedom geometrically, we observe that the exterior derivative $c:=d\widehat j$
does not dependent on the choice of the phase $\varphi$, i.e.,
\begin{align}
    d\widehat j_A(G,F)
    =:c_A(G, F) = d \widehat j_A (G, F)
    =
    \partial_G \widehat j_{A+G}(F)
    - \partial_F \widehat j_{A+F}(G)
    =d\widetilde j_A(G,F)
    \label{eq_dj}
\end{align}
because $dd\varphi=0$. By Poincaré's lemma and the fact that the space of permissible vector potentials $A$ is
star-shaped, the two-form $c$ contains precisely the same information as the class of all $\widehat j$ with varying
$\varphi$. By construction, the two-form $c$ is closed, i.e., $d c = dd\widehat
j = 0$. This will play a crucial role in
the non-perturbative construction addressed in forthcoming papers. 

The physically relevant $j$ should now be selected with conditions 
\ref{Exterior derivative}-\ref{Reference current} in mind: 
\begin{enumerate}[start=0,label=C\arabic*]
    \item\label{Exterior derivative} \emph{Exterior derivative}: Given \eqref{eq_dj}, we require $dj=c$; 
    \item\label{Causality} \emph{Causality:} For $F,G\in\cC^\infty_c(\R^4,\R^4)$ such that
      the support of $G$
      does not overlap the closed causal past of the support of $F$, we
        require $\partial_G j_{A+G}(F)=0$;
    \item\label{Relativistic invariance} \emph{Relativistic invariance:} For any proper, orthochronous
        Lorentz-transformation $\Lambda$ and any translation displacement $a\in\R^4$, 
        we require $j_{\Lambda A}(\Lambda F)=j_A(F)=j_{A(\cdot -a)}(F(\cdot-a))$;
    \item\label{Gauge invariance} \emph{Gauge invariance:} For $\Gamma\in\cC^\infty_c(\R^4,\R)$,
        we require $j_{A+\partial \Gamma}(F)=j_A(F)=j_A(F+\partial\Gamma)$;
    \item\label{Reference current} \emph{Reference current:} Should $\Aref$ be
        sufficiently close to zero to allow for $A=0$, then we require the vacuum
        expectation of the current to vanish in this case, i.e., $j_{0}(F)=0$.
\end{enumerate}

In order to derive an explicit expression for $c$ for \ref{Exterior derivative}, cf.~\eqref{eq_dj}, we start by computing $\widetilde j_A(F)$
in~\eqref{eq_current}. By~\eqref{eq_vacuum_exp},~\eqref{eq_right_op}, and~\eqref{eq_lift} we find
\begin{align}
    \widetilde{j}_A(F)
    =&
    \re
    i\partial_F
    \det \left[
        {(\sA^\inv \sAf)}_\Pmm
        (\sAf^\inv \sB)_\Pmm
        ((\sA^\inv\sB)_\Pmm)^\inv
        \right]
    \notag
    \\
    &\times \det |(\sAf^\inv \sB)_\Pmm|^\inv
    \det |(\sA^\inv\sB)_\Pmm|
    \notag
    \\
    =&
    \re
    i\partial_F
    \det \left[
        (\sA^\inv \sAf)_\Pmm
        (\sAf^\inv \sB)_\Pmm
        (\sB^\inv\sA)_\Pmm
        \right]
    \notag
    \\
    &\times \det |(\sAf^\inv \sB)_\Pmm|^\inv
    \det |(\sA^\inv\sB)_\Pmm|^\inv
    \,,
\end{align}
where we have used $|U_\Pmm| (U_\Pmm)^\inv  = |U_\Pmm|^\inv (U^\inv)_\Pmm$ for $U=\sA^\inv\sB$ and have performed a cyclic permutation
under the determinant by a corollary of Lidskii’s theorem~\cite{simon2005trace}. For $F=0$ the product of the three
determinants equals one because $\det (\sA^\inv \sA)_\Pmm (\sA^\inv \sB)_\Pmm (\sB^\inv\sA)_\Pmm$ equals $\det
|(\sB^\inv\sA)_\Pmm|^2$. This allows to recast the above expression as $\widetilde{j}_A(F) =i\partial_F \log
\Gamma_{\sAf \; \sB \; \sA}$, using the notation $\Gamma_{XYZ}:=\arg\det [(Z^\inv X)_\Pmm(X^\inv Y)_\Pmm (Y^\inv
Z)_\Pmm]$. These terms $\Gamma_{XYZ}$ have convenient properties summarised in the appendix in Lemma~\ref{lem_Gamma}.
Exploiting those and the chain rule, the exterior derivative $c=d\widetilde j$, cf.~\eqref{eq_dj}, can be expressed as
follows, using the notations $S^{X}_{Y}:=S_X^\inv S_{Y}=(S^Y_X)^\inv$
and $\arg z:=z/|z|$ for $z\in\C\setminus\{0\}$:
\begin{align}
    c_A(G,F)
    &=
    i\partial_G \partial_F \log \Gamma_{\sAfg \; \sB \; \sAg}
    - i\partial_F \partial_G \log \Gamma_{\sAfg \; \sB \; \sAf}
    \notag
    \\
    &= 2i \partial_F\partial_G \log \Gamma_{\sAf\; \sB \; \sAg}
    \qquad \text{by Lem.~\ref{lem_Gamma}, prop.~2.}
    \notag
    \\
    &= 2i \partial_F\partial_G \log \Gamma_{\sAf\; \sA \; \sAg}
    \hskip1.1cm \text{by Lem.~\ref{lem_Gamma}, prop.~4.}
    \notag
    \\
    &= 2i \partial_F\partial_G \log \arg \det
    [
        (S^{A+G}_{A+F})_\Pmm (
        (S^{A+F}_{A+G})_\Pmm
        -
        (S^{A+F}_A)_\Pmp
        (S^{A}_{A+G})_\Ppm
        )
    ]
    \notag
    \\
    &= 2i \partial_F\partial_G \log \arg \det [
        \id_\Pmm - |(S^{A+G}_{A+F})_\Ppm|^2
        - (S^{A+G}_{A+F})_\Pmm 
        (S^{A+F}_{A})_\Pmp 
        (S^{A}_{A+G})_\Ppm 
    ]
    \notag
    \\
    &= -2 \partial_F\partial_G \im \Tr [ (S^A_{A+F})_\Pmp (S^A_{A+G})_\Ppm ]
    \,
      \label{eq_tracewurm}
\end{align}
Recalling~\eqref{eq_dj}, in order to get our hands on a physically relevant
candidate current $j$ in the sense of above, we need to split $c$ as follows
\begin{align}
    c_A(G, F) =
    \partial_G j_{A+G}(F) - \partial_F j_{A+F}(G)
    \,.
    \label{eq_non_perturb_splitting}
\end{align}
Note that the left-hand side is given as the non-perturbative expression~\eqref{eq_tracewurm}. For the purpose of the
splitting, we make the following ansatz. Suppose the vacuum expectation value of the current $A\mapsto j_A$ is real
analytic and has a power expansion of the form
\begin{align}
    j_A(F)  
    = 
    \sum_{n=2}^\infty  j^{(n)}(F; \underbrace{A,\dots, A}_{n-1\text{ many}}) \, ,
    \label{eq_j_series}
\end{align}
where the $n$-th summand on the right-hand side is assumed to be 
linear in all $n$ arguments and symmetric in the last $n-1$ arguments.
The same summand corresponds to the $(n-1)$-st Taylor order in $A$. In view of condition \ref{Reference current} there is no
$n=1$ summand. 
This implies the expansion
\begin{align}
    c_A(G,F) &= \sum_{n=2}^\infty c^{(n)}(G,F;\underbrace{A,\dots, A}_{n-2\text{ many}}) \, , \qquad \text{for}
    \\
    \label{eq_splitting_equation}
    c^{(n)}(G,F;A,\dots, A)
    &=
    (n-1)
    \big(
    j^{(n)}(F;G,A,\dots, A)
    -
    j^{(n)}(G;F,A,\dots, A)
    \big) \, .
\end{align}
In forthcoming works, we will justify the ansatz, i.e., the analyticity assumption, provide a non-perturbative form
for the current and, on its basis, construct the corresponding scattering matrix and time-evolution. In this article, we
will only demonstrate how to perform the splitting \eqref{eq_non_perturb_splitting} for the lowest order term $n=2$, in order
to fulfil \ref{Exterior derivative}, and check the remaining \ref{Causality}-\ref{Reference current}.

\section{Second order perturbation without infinities}

In view of \ref{Exterior derivative}, cf.~\eqref{eq_splitting_equation} for the lowest order $n=2$, we need
to find an expression $j^{(2)}$ fulfilling
\begin{align}
    c^{(2)}(G,F)
    =
    j^{(2)}(F;G)
    -
    j^{(2)}(G;F)
    \label{eq_splitting_2}
\end{align}
together with conditions \ref{Causality}-\ref{Reference current}. Not
surprisingly, it will coincide with the well-known expression for the
second order perturbation of the current of QED; e.g.~\cite{Dyson2006}. However, in text-books, the latter is extracted
from a mathematically non-rigorous computation involving infinities that are removed by hand. In what follows, we give
this computation a mathematical sense.

We obtain the following finite Lebesgue integral for $c^{(2)}(G,F)=c_0(G,F)$,
cf.~\eqref{eq_tracewurm}, after employing \eqref{eq_comparison_dyn} in order to
compute the linearisation of $S_{A+F}^A$, and furthermore, a suitable partial integration in
time to allow for the application of Fubini's theorem:
\begin{align}
    &
    \hskip-0.25cm
    c^{(2)}(G,F) 
    = 
    i \, \partial_G \partial_F \Tr \left(
    P_- S^A_{A+F} P_+ S^A_{A+G} P_-
    -
    P_- S^A_{A+G} P_+ S^A_{A+F} P_-
    \right)
    {=}
    \hskip-0.3cm
    \int\limits_{\R^3\times\R^3} \!\!\!\! d^3p \, d^3q \,
    c^{(2)}(G,F;\mathbf q,\mathbf p)
    \nonumber
    \\
    &
    \text{with }\; c^{(2)}(G,F;\mathbf q,\mathbf p)
    :=
    \sum_{\tau=\pm 1} \tau \cdot
    (2\pi i)^2 \Res_{\substack{p^0=-\tau E(\mathbf p),\\q^0=\tau E(\mathbf q)}}
    \omega(G,F;q,p)\,,\;
    E(\mathbf{p}):=\sqrt{m^2+\mathbf{p}^2},
    \label{eq_c2_kernel}
    \\
    &
    \omega(G,F;q,p) 
    :=
    \frac{i e^2}{(2\pi)^4} 
    \tr [(\slashed p-m)^{-1} \slashed{F}(p-q) (\slashed q-m)^{-1} \slashed{G}(q-p)]
    =
    \omega(F,G;p,q)
    \, .
    \nonumber
\end{align}
By abuse of notation, we denote the Fourier transform of $x\mapsto F_\mu(x)$ by the same symbol $k\mapsto
F_\mu(k)=(2\pi)^{-4/2}\int d^4x\, e^{ik_\nu x^\nu} F_\mu(x)$. Furthermore, $\Res_{q^0=\cdot,p^0=\cdot}$ denotes the
iterated residue operator.

In order to identify the current term $j^{(2)}$ in \eqref{eq_splitting_2}, we first regard the expression
$\omega(q^0,p^0)$ for fixed $\mathbf q,\mathbf p\in\R^3$. In a first step, we add zero, written as a difference of two
equal residues, i.e.,
\begin{align}
    &
    c^{(2)}(G,F;\mathbf q,\mathbf p)
    =\sum_{\sigma=\pm 1}
    (2\pi i)^2 
    \left(
        \Res_{\substack{p^0=- E(\mathbf p),\\q^0= \sigma E(\mathbf q)}}
        -
        \Res_{\substack{p^0=\sigma E(\mathbf p),\\q^0= - E(\mathbf q)}}
    \right)
    \omega(G,F;q,p)
    \nonumber
    \\
    &=
    (2\pi i)
    \Bigg(
        \int\limits_{[\R-i\delta]-[\R+i\delta]}dq^0\,
        \Res_{p^0=-E(\mathbf p)}
        -
        \int\limits_{[\R-i\delta]-[\R+i\delta]}dp^0\,
        \Res_{q^0= -E(\mathbf q)}
    \Bigg)
    \omega(G,F;q,p)
    \nonumber
    \\
    &=
    c^{(2)}_+(G;F;\mathbf q,\mathbf p)
    -
    c^{(2)}_+(F;G;\mathbf p,\mathbf q)
    \label{eq_splitting}
\end{align}
for any fixed number $\delta>0$ and
\begin{align}
    c^{(2)}_+(F;G;\mathbf p,\mathbf q)
    :=
    -2\pi i
    \Bigg(
    \int\limits_{[\R+i\delta]}dq^0\,
    \Res_{p^0= - E(\mathbf p)}
    +
    \int\limits_{[\R-i\delta]}dp^0\,
    \Res_{q^0= - E(\mathbf q)}
    \Bigg)
    \omega(F,G;p,q) \,.
    \label{eq_c2_eq_diffj}
\end{align}
The bracket notation $[\R\pm i\delta]$ denotes the standard parametrisation $\R\ni t\mapsto t\pm i\delta$ and their
differences are understood as $\int_{[A]-[B]}=\int_{[A]}-\int_{[B]}$. We remark that $c^{(2)}_+$ has temporally causal
support in the sense that, for $F,G\in\cC^\infty_c(\R^4,\R^4)$ such that the support of $F$ is temporally earlier than
the support of $G$, we have $c^{(2)}_+(F;G;\mathbf p,\mathbf q)=0$. This can be seen by inserting the Fourier
transformations of $F$ and $G$ in the time variable, observing that $\C\setminus\R\ni p^0\mapsto (\slashed p-m)^{-1}$ is
holomorphic, and applying Cauchy's integral theorem in the limit $\delta\to+\infty$. Furthermore, it is convenient to
introduce the substitution $k=(k^0,\mathbf k):=p-q$, i.e.,
\begin{align}
    &c^{(2)}_+(F;G;\mathbf p,\mathbf q)
    =
    -2\pi i
    \int\limits_{[\R-i\delta]} dk^0 \,
    \left(
       \Res_{p^0=-E(\mathbf p)}
       +\Res_{p^0=k^0-E(\mathbf p-\mathbf k)}
    \right) \omega(F,G;p,p-k) 
    \label{eq_c2plus}
    \\
    &=
    - \!\! \int\limits_{[\R-i\delta]} \!\! dk^0 \!\! \int\limits_{\cC_{\text{Wick}}(k)} \!\! dp^0 \,
    \omega(F,G;p,p-k) 
    =
    -\frac{ie^2}{(2\pi)^4} 
    \int\limits_{[\R-i\delta]} dk^0 
    F_\mu(k)G_\nu(-k)
    \int\limits_{\cC_{\text{Wick}}(k)} dp^0 \,
    \omega^{\mu\nu}(p,k)
    \nonumber
    \\
    &\hskip3cm \text{with} \qquad \omega^{\mu\nu}(p,k)
    :=
    \tr\left[
        \gamma^\nu
        (\slashed p-m)^{-1} 
        \gamma^\mu
        (\slashed p-\slashed k-m)^{-1} 
        \right] \,.
    \label{eq_Pi_pk}
\end{align}
Here, the contour $\cC_{\text{Wick}}(k)$ denotes any closed curve having winding number $+1$ around $p^0=-E(\mathbf p)$
and $p^0=k^0-E(\mathbf p-\mathbf k)$, but winding number zero around $p^0=+E(\mathbf p)$ and $p^0=k^0+E(\mathbf
p-\mathbf k)$. We remark that the inner integral of \eqref{eq_c2plus}, for fixed $\mathbf p,\mathbf k$ and taken as function of
$k^0$, is holomorphic on 
\begin{align}
    k^0\in \mathbb{D}:=\C\setminus((-\infty,-2m] \cup [2m,\infty)).
    \label{eq_hol_set}
\end{align}
In particular, for $k\in\C^4$ sufficiently
close to zero, $\omega^{\mu\nu}(p,k)$ is well-defined for all $p\in i\R\times\R^3$. In this region the contour
$\cC_{\text{Wick}}(k)$ can be replaced by $i\R$ oriented in positive imaginary direction, exploiting the 
$|p^0|^{-2}$ decay of $\omega^{\mu\nu}(p,k)$ for $|p^0|\to\infty$.

Considering~\eqref{eq_splitting_2}, if $c^{(2)}_+(F;G;\mathbf p,\mathbf q)$ was integrable in $(\mathbf p,\mathbf
q)\in\R^3\times\R^3$, its integral would have been a natural candidate for $j^{(2)}(F;G)$ due to the support
properties discussed above. But this is in general not the case. To nevertheless find a suitable candidate by leveraging the knowledge of
$c^{(2)}_+$, we take a different approach: It will turn out that second derivative
$\partial^2_{m^2}c^{(2)}_+(F;G;\cdot,\cdot)$ belongs to $L^1$, and therefore, with a function $\Pi^{\mu\nu}$ to be
found, we shall search instead for a candidate of the form
\begin{align}
    &j^{(2)}(F;G)
    =
    -\frac{ie^2}{(2\pi)^4} 
    \int\limits_{[\R-i\delta]\times\R^3} d^4k \, F_\mu(k) G_\nu(-k) \Pi^{\mu\nu}(k) \,,
    \label{eq_def_j2}
    \\
    &\hskip3cm\text{fulfilling}
    \quad
    \partial^2_{m^2}
    j^{(2)}(F;G) 
    =
    \int\limits_{\R^3\times\R^3} d^3p \, d^3q \,
    \partial^2_{m^2}
    c^{(2)}_+(F;G;\mathbf p,\mathbf q)
    \,.
    \label{eq_def_j2_2}
\end{align}
Throughout the article, we have suppressed the $m$ dependence in the notation. In the end, we shall integrate twice with
respect to $m^2$; taking the second derivative with respect to $m^2$ rather than $m$ is only a technical convenience.
This can be seen as a variant of the Pauli-Villars regularisation scheme \cite{PauliVillars49} with differences replaced by integrals of
derivatives. The expression for $\Pi^{\mu\nu}$ will finally be identified as \eqref{eq_def_pi_munu}. 

\section{The explicit expression for the second order}

Having in mind the goal~\eqref{eq_def_j2_2} with $c^{(2)}_+$ as given in~\eqref{eq_c2plus} we observe that
even for $k$ close to zero, $\omega^{\mu\nu}(p,k)$ in~\eqref{eq_Pi_pk} is not Lebesgue-integrable in
$p\in i\R\times\R^3$. However, we have the Lebesgue integrals
\begin{align}
    \label{eq_Pi_integral}
    \int\limits_{\mathcal{C}_{\text{Wick}}(k)\times\R^3} d^4p \,
    \partial_{m^2}^2
    \omega^{\mu\nu}(p,k)
    \overset{\text{for $k$}}{\underset{\text{close to $0$}}{=}}
    \int\limits_{i\R\times\R^3} d^4p \,
    \partial_{m^2}^2
    \omega^{\mu\nu}(p,k)\,.
\end{align}
Note that by dominated convergence, the right-hand side in \eqref{eq_Pi_integral} is a real-analytic function of
$\mathbf{k}\in\R^3$ and a holomorphic function of $k^0\in\mathbb{D}$; cf.~\eqref{eq_hol_set}.  

In order to avoid the explicit computation of the second derivative, we
introduce an artificial scaling parameter $\epsilon$ on which the integral shall depend meromorphically. The original
integral~\eqref{eq_Pi_integral} is then recovered for $\epsilon\to 0$. However, by virtue of the identity theorem for
analytic functions, the value at $\epsilon=0$ is determined by values for $\epsilon$ in any non-empty, open interval $I$
of reals, which need not be in the vicinity of zero. It will turn out that $I$ can be chosen so that the differential
operator $\partial_{m^2}^2$ can be interchanged with the integral by dominated convergence without losing
Lebesgue-integrability for $\epsilon\in I$. A scaling exponent $\epsilon$ can be introduced conveniently after recasting
the integral by means of Feynman's parametrisation. This computation procedure can be seen as a mathematically rigorous
version of the method of dimensional regularisation employed in physics \cite{THOOFT1972189,Bollini72}. We interpret the
attribute ``dimensional'' as a scaling exponent $\epsilon$, which may be any real or even complex number, but not as the
number of elements in a basis of a vector space, which would be a natural number. Next, we shall demonstrate this
procedure for the second order term, which can be seen as a rigorous version of the computation in \cite[p.~70ff:
Polarization of the Vacuum]{Dyson2006}.

In order to arrive at an explicit expression for the second perturbation order
of the current, we employ the Feynman parametrisation $(ab)^{-1}=\int_0^1
((1-z)a+zb)^{-2}\,dz$, which holds for all $a,b\in\C$ such that the straight
line from $a$ to $b$ does not contain $0$. For $k\in\C^4$ sufficiently close to
zero, $p\in i\R\times\R^3$ we may take $a=p^2-m^2$ and
$b=(p-k)^2-m^2$ and recast~\eqref{eq_Pi_integral} into
\begin{align}
    & 
    \int\limits_{i\R\times\R^3} \!\!\!\!\! d^4p \,
    \partial_{m^2}^2
    \frac{
        \tr[
            \gamma^\nu (\slashed p+m) \gamma^\mu (\slashed p-\slashed k + m)
        ]}
        {
            (p^2-m^2)((p-k)^2-m^2)
        }
        =\!\!\!\!\!
        \int\limits_{i\R\times\R^3} \!\!\!\!\! d^4p \,
        \partial_{m^2}^2
        \!\!\int\limits_0^1 \!\! dz
        \frac{
            \tr[
                \gamma^\nu (\slashed p+m) \gamma^\mu (\slashed p-\slashed k + m)
            ]}
            {
                \left[(1{-}z)(p^2-m^2) + z((p-k)^2-m^2)\right]^2
            }
            \notag
            \\
            &=
            \int\limits_0^1 dz
            \int\limits_{i\R\times\R^3} d^4p \,
            \partial_{m^2}^2
            \frac{
                \tr[
                    \gamma^\nu (\slashed p+m) \gamma^\mu (\slashed p-\slashed k + m)
                ]}
                {
                    \left[(1{-}z)(p^2-m^2) + z((p-k)^2-m^2)\right]^2
                }\,,
\end{align}
where in the last step we have used dominated convergence to interchange
$\partial^2_{m^2}$ and $\int dz$ and the Lebesgue-integrability to
interchange the integrals. For the sake of the computation, we may restrict
ourselves to $k\in i\R\times\R^3$; the rest of the computation does not even
require $k$ to be close to zero. Next, employing the substitution $q=p-kz$,
which is $m$-independent, to find 
\begin{align}
    {\eqref{eq_Pi_integral}}
    =
    \int\limits_0^1 dz
    \int\limits_{i\R\times\R^3} d^4q \,
    \partial_{m^2}^2
    \frac{\tr[
        \gamma^\nu(\slashed q+z\slashed k+m)
        \gamma^\mu(\slashed q + (z{-}1)\slashed k +m)
    ]}
    {
        [m^2 - (z{-}z^2)k^2 - q^2]^2\,
    },
\end{align}
and evaluating the trace, using $\tr \gamma^\nu\gamma^\mu = 4g^{\mu\nu}$,
$\tr\gamma^\nu\gamma^\mu\gamma^\sigma=0$, and
$\tr\gamma^\nu\gamma^\sigma\gamma^\mu\gamma^\tau =
4(g^{\mu\sigma}g^{\nu\tau}-g^{\mu\nu}g^{\sigma\tau}+g^{\mu\tau}g^{\sigma\nu})$,
as done, e.g., in~\cite[Equations (377)-(379)]{Dyson2006}, and dropping terms of the
form $\int d^4q \, q_\mu k^\mu f(q^2)=0$, results in
\begin{align}
    &{\eqref{eq_Pi_integral}}
    =
    \int\limits_0^1 dz
    \int\limits_{i\R\times\R^3} d^4q 
    \,
    \partial_{m^2}^2
    f^{\mu\nu}_{k,m}(q,z)
    \\
    &\text{with}
    \qquad
    f^{\mu\nu}_{k,m}(q,z)
    :=
    4\frac{
        2q^\mu q^\nu - 2k^\mu k^\nu(z{-}z^2) + g^{\mu\nu} (m^2+k^2(z{-}z^2)-q^2) 
    }{
        [m^2 - (z{-}z^2)k^2 - q^2]^{2}
    }
    \,.
\end{align}
Due to the Minkowski inner-product, we have $-k^2\geq 0$, $-q^2\geq 0$. Note further
that $f^{\mu\nu}_{q,k,z,m}$ behaves as $O_{|q|\to\infty}(|q|^{-2})$ and
$O_{|q|\to 0}(1)$ uniformly for $z\in[0,1]$, $k$ in a compact domain in
$i\R\times\R^3$, and also $m$ in a compact domain in $\R^+$.

As announced above, we will now introduce an artificial scaling
$(-q^2/u)^\epsilon$ with a complex exponent $\epsilon$ and another parameter
$u$ carrying the unit of $m^2$; recall that $\hbar=1=c$. The original
term~\eqref{eq_Pi_integral} shall be retrieved in the limit $\epsilon\to 0$.
\begin{align}
    {\eqref{eq_Pi_integral}}
    =
    \lim_{\epsilon\to 0}
    \int\limits_0^1 dz
    \int\limits_{i\R\times\R^3} d^4q 
    \,
    \big(\frac{-q^2}{u}\big)^\epsilon 
    \,
    \partial_{m^2}^2
    f^{\mu\nu}_{k,m}(q,z)
    \label{eq_lim_Pi_integral}
\end{align}
Note that the integrand is an Lebesgue-integrable, holomorphic function of
$\epsilon$ for $-2<\re\epsilon<+1$ which allows to interchange the limit as
$\epsilon\to 0$ with the integral by dominated convergence. Even when dropping
the $\partial^2_{m^2}$ operator, it remains so for the smaller domain $-2<\re
\epsilon < -1$, which does, however, not contain zero any more. First, we will
ignore this complication in the computations and work with smaller domain but,
in the end, recover a domain that contains $\epsilon=0$ by analytic
continuation. On the smaller domain, using dominated convergence once again, we
find
\begin{align}
    \int\limits_0^1 dz
    \int\limits_{i\R\times\R^3} d^4q 
    \,
    \big(\frac{-q^2}{u}\big)^\epsilon 
    \partial_{m^2}^2
    f^{\mu\nu}_{k,m}(q,z)
    =
    \int\limits_0^1 dz\,
    \partial_{m^2}^2
    \int\limits_{i\R\times\R^3} d^4q 
    \,
    \big(\frac{-q^2}{u}\big)^\epsilon 
    f^{\mu\nu}_{k,m}(q,z)\,.
    \label{eq_exchange_int_partial_m}
\end{align}
We will now evaluate the inner integral for $-2<\re \epsilon<-1$. 
The following integral 
\begin{align}
    I(\zeta, \eta) := 
    {u^{-\eta}} \int\limits_{i\R\times\R^3} d^4q
    \frac{(-q^2)^\eta}{(\zeta u -q^2)^2}
    =
    \int\limits_{i\R\times\R^3} \frac{d^4q}{u^2}
    \frac{\big(\frac{-q^2}{u}\big)^\eta}{(\zeta-q^2/u)^2}
    \,,
    \quad
    \text{for }
    \zeta>0, 
    \re \eta \in (-2,0)
    \,,
    \label{eq_I_integral}
\end{align}
will play an important role in these calculations. The decoration
by factors of powers of $u$ renders the expression~\eqref{eq_I_integral}
dimensionless. By euclidean symmetry of the complexified Minkowski
inner-product on $i\R\times\R^3$ and for $\zeta>0$, we have
\begin{align}
    \int\limits_{i\R\times\R^3} d^4q 
    \,
    \big(\frac{-q^2}{u}\big)^\epsilon
    \frac{
        q^\mu q^\nu
    }{
        [\zeta u - q^2]^2
    }
    =
    \int\limits_{i\R\times\R^3} d^4q 
    \,
    \big(\frac{-q^2}{u}\big)^\epsilon
    \frac{
        \frac14 g^{\mu\nu} q^2
    }{
        [\zeta u - q^2]^2
    }
    = -\frac u 4 g^{\mu\nu} I(\zeta, 1+\epsilon)
    \,,
\end{align}
Using~\eqref{eq_exchange_int_partial_m},~\eqref{eq_I_integral}, and
$\zeta:=(m^2-(z-z^2)k^2)/u$, we recast the inner integral of~\eqref{eq_lim_Pi_integral} into
\begin{align}
    F^{\mu\nu}_{k,m}(z,\epsilon) 
    &:=
    \int\limits_{i\R\times\R^3} d^4q 
    \,
    \big(\frac{-q^2}{u}\big)^\epsilon 
    f^{\mu\nu}_{k,m}(q,z)
    \label{eq_lin_comb_I}
    \\
    &=
    2 g^{\mu\nu} u I(\zeta,1+\epsilon) + 4[(g^{\mu\nu}k^2-2k^\mu
    k^\nu)(z{-}z^2)+g^{\mu\nu}m^2] I(\zeta,\epsilon) \,.
    \nonumber
\end{align}
Next, we evaluate the integral $I(\zeta,\eta)$ for $\zeta>0$ and $-2 < \re
\eta < 0$, i.e.,
\begin{align}
    I(\zeta,\eta) 
    &= 2i\pi^2 u^{-\eta}\int\limits_0^\infty \frac{r^{3+2\eta}}{(\zeta u + r^2)^2} dr
    = i\pi^2 \zeta^\eta \operatorname{B}(\eta+2,-\eta)  
    = i\pi^2 \zeta^\eta \Gamma(2+\eta)\Gamma(-\eta)
    \notag
    \\
    &= i\pi^2 \zeta^\eta (1+\eta) \Gamma(1+\eta)\Gamma(-\eta)
    = i\pi^2 \zeta^\eta \frac{\pi (1+\eta)}{\sin(\pi(1 + \eta))}
    \label{eq_I_integral_euler}
    \,,
\end{align}
where $\operatorname{B}$ denotes the beta function and, in the last step, Euler's
reflection formula was used. For all given $\zeta>0$, the right-hand side of~\eqref{eq_I_integral_euler} implies that the function
$(-2,-1)+i\R\ni\eta\mapsto I(\zeta,\eta)$ has a meromorphic extension on the
whole complex plane with poles of first order at most on $\Z\setminus\{-1\}$. We denote
it by the same symbol $I$. Hence, $F^{\mu\nu}_{k,m}(z,\epsilon)$ in~\eqref{eq_lin_comb_I} also has a holomorphic extension for
$\epsilon\in\C\setminus\Z$ which will again be denoted by the same symbol. It
is important to note that this extends the originally much smaller domain of $\epsilon\in
(-2,-1) + i\R$, which did not even include a neighbourhood of $\epsilon=0$.
Given $\zeta=(m^2-(z{-}z^2)k^2)/u>0$, this extension now allows to expand for
$\epsilon\to0$ as follows:
\begin{align}
    &I(\zeta, 1+\epsilon) 
    = 
    2i\pi^2 \zeta \left(
    \frac1\epsilon + \log\zeta + \frac12 + O^\zeta_{\epsilon\to 0}(\epsilon)
    \right)\,, \;
    I(\zeta, \epsilon) 
    = 
    -i\pi^2 \left(
    \frac1\epsilon + \log\zeta + 1 + O^{\zeta}_{\epsilon\to 0}(\epsilon)
    \right)\,,
    \nonumber\\
    &F^{\mu\nu}_{k,m}(z,\epsilon) 
    = 
    4  
    i\pi^2  [g^{\mu\nu} m^2-g^{\mu\nu} k^2(z{-}z^2)] \left(
    \frac1\epsilon + \log\zeta + \frac12 
    \right)
    \nonumber\\
    &
    \hskip2.5cm - 4i\pi^2 [g^{\mu\nu}m^2+(g^{\mu\nu}k^2-2k^\mu k^\nu)(z{-}z^2)] 
    \left(
    \frac1\epsilon + \log\zeta + 1 
    \right)
    +
    O^{z,k,m}_{\epsilon\to 0}(\epsilon)
    \nonumber\\
    &=
    8i\pi^2(k^\mu k^\nu-g^{\mu\nu}k^2)(z{-}z^2)
    \left(\frac{1}{\epsilon}+\log\zeta\right)
    +2i\pi^2[4k^\mu k^\nu-g^{\mu\nu}(m^2+3k^2)](z{-}z^2)
    +
    O^{z,k,m}_{\epsilon\to 0}(\epsilon)      
    \,.
    \nonumber
\end{align}
We emphasise that the remainder $O^{z,k,m}_{\epsilon\to 0}(\epsilon)$ is
uniform in the parameters $\zeta$ and $k,m,z$, respectively, as long as they
are restricted to compact domains. Moreover, it is smooth in $m$ with
$\partial^n_{m^2}O^{z,k,m}_{\epsilon\to 0}(\epsilon)=O^{z,k,m}_{\epsilon\to
0}(\epsilon)$, $n\in\N$. Hence, the expressions in~\eqref{eq_lin_comb_I}
and, exploiting uniformity and smoothness of the remainder term, also
the expression in~\eqref{eq_exchange_int_partial_m}, taken
on the smaller domain $-2<\re\epsilon<-1$,
has a
holomorphic extension for $\epsilon\in\C\setminus\Z$ given by
\begin{align}
    &\int\limits_0^1 dz\,\partial_{m^2}^2F^{\mu\nu}_{k,m}(z,\epsilon)
    =
    \partial_{m^2}^2
    \Pi^{\mu\nu}(k)
    +O^{k,m}_{\epsilon\to 0}(\epsilon)
    \label{eq_hol_ext}
    \quad \text{with} \\
    &\qquad
    \Pi^{\mu\nu}(k):=
        8i\pi^2(k^\mu k^\nu-g^{\mu\nu}k^2)\int\limits_0^1 dz\,(z{-}z^2)
        \left[
            \log\left(1-(z{-}z^2)\frac{k^2}{m^2}\right)+\log\frac{m^2}{u}
        \right]
    \label{eq_def_pi_munu}
    \, . 
\end{align}
The last expression shows that the isolated singularity at $\epsilon=0$ is removable. Applying the identity theorem for
holomorphic functions guarantees that the left-hand side of~\eqref{eq_exchange_int_partial_m} is given by the expression on
the left-hand side in~\eqref{eq_hol_ext} for $-2<\epsilon<+1$, which is a neighbourhood of $\epsilon=0$. Note that thanks
to the derivative $\partial_{m^2}^2$, the term is independent of $u$, as it must be. Furthermore, the
$1/\epsilon$-singularity has dropped out thanks to the same derivative, which is consistent with the fact that the
left-hand side of~\eqref{eq_exchange_int_partial_m} is holomorphic near $\epsilon=0$.  Evaluating the holomorphic
extension at $\epsilon=0$, we have verified that the quantity $\Pi^{\mu\nu}(k)$ defined in \eqref{eq_def_pi_munu}
fulfils 
\begin{align}
    \partial_{m^2}^2 \Pi^{\mu\nu}(k) = 
    \int\limits_{\mathcal{C}_{\text{Wick}}(k)\times\R^3} d^4p \,
    \partial_{m^2}^2
    \omega^{\mu\nu}(p,k)
    \label{eq_pi_munu}
\end{align}
for $k\in i\R\times\R^3$. Interpreting the logarithm in \eqref{eq_def_pi_munu} as its principal value
$\log:\C\setminus\R^-_0\to\C$, the function $\Pi^{\mu\nu}$ on $i\R\times\R^3$ has a holomorphic extension to
$\{k\in\C^4|\,k^2\notin\R\text{ or } k^2<4m^2\}$ which contains a neighbourhood of $0\in\C^4$ as well as the sets $\mathbb{D}\times\R^3$ and
$\R^4+i\overline{\operatorname{Causal}}$, where
$\overline{\operatorname{Causal}}=\{p\in\R^4:\,p^2\ge 0\}$ denotes the set of time-like or light-like four vectors.
This implies \eqref{eq_pi_munu} also for $k\in\mathbb{D}\times\R^3$ by virtue of the identity theorem for analytic functions. Note further that $\Pi^{\mu\nu}(k)$ is of the order
$O^{m,u}_{|k|\to\infty}(|k|^2 \log|k|)$ and thus allows the integral on the right-hand side of \eqref{eq_def_j2} to be
well-defined and independent of $\delta>0$.

It is left to show \eqref{eq_splitting_2}. For this purpose, we claim
\begin{align}
    \partial^2_{m^2}
    \left(
        j^{(2)}(F;G)
        -
        j^{(2)}(G;F)
        -
        c^{(2)}(G,F)
    \right)
    =
    0
    \,.
    \label{eq_splitting_derivative}
\end{align}
To see this, we recall~\eqref{eq_def_j2} and regard
\begin{align}
    &\partial^2_{m^2}
    j^{(2)}(F;G) 
    \overset{\eqref{eq_pi_munu}}{=}
    -
    \!\!\!\!
    \int\limits_{[\R-i\delta]\times\R^3} d^4 k\,
    F_\mu(k)G_\nu(-k) 
    \int\limits_{\mathcal{C}_{\text{Wick}}(k)\times\R^3} d^4p \,
    \partial_{m^2}^2
    \omega^{\mu\nu}(p,k)
    \\
    &=
    -
    \!\!\! \int\limits_{\R^3\times\R^3} \!\!\! d^3k d^3p \,
    \partial_{m^2}^2
    \!\!\! \int\limits_{[\R-i\delta]} \!\!\! dk^0 \,
    F_\mu(k)G_\nu(-k) 
    \!\!\! \int\limits_{\mathcal{C}_{\text{Wick}}(k)} \!\!\! dp^0 \,
    \omega^{\mu\nu}(p,k)
    =
    \!\!\! \int\limits_{\R^3\times\R^3} \!\!\! d^3p \, d^3q \,
    \partial^2_{m^2}
    c^{(2)}_+(F;G;\mathbf p,\mathbf q)\,,
    \nonumber
\end{align}
where the employed commutation of the integrals and differential operators is justified by the integrability of the
integrand, locally uniform in $m$. Inserting equation \eqref{eq_splitting} proves the claim \eqref{eq_splitting_derivative}.
Furthermore, \eqref{eq_splitting_derivative} implies
\begin{align}
    j^{(2)}(F;G)
    -
    j^{(2)}(G;F)
    -
    c^{(2)}(G,F)
    =
    a m^2 + b,
\end{align}
with two constants $a,b$ that depend on $F,G$ but are independent of $m$. Note that $j^{(2)}(F;G) - j^{(2)}(G;F)\to 0$
for $m\to\infty$ because the scaling term $\log(m^2/u)$ cancels in the difference; see \eqref{eq_def_pi_munu}. Moreover,
we have $c^{(2)}(G,F)\to 0$ for $m\to\infty$, which can be seen from \eqref{eq_c2_kernel}, 
the estimate $|p-q|\geq |p^0-q^0| = E(\mathbf
p)+E(\mathbf q)\geq 2^{-1/2}(|\mathbf p|+|\mathbf q|+2m)$ for $p^0=-\tau E(\mathbf p)$, $q^0=\tau E(\mathbf q)$, and
$\tau=\pm 1$, implied by Cauchy-Schwarz' inequality, and noting the fact that $F,G$ decay
super-algebraically fast in energy-momentum space. This guarantees $a=0=b$, i.e., \eqref{eq_splitting_2}. 

\paragraph{Conclusion.} Up to second order in perturbation, we have retrieved a physically relevant family of currents
by~\eqref{eq_def_j2} and~\eqref{eq_def_pi_munu}, depending on an integration constant $C\in\R$, i.e., 
\begin{align}
    j^{(2)}(F;A)
    &=
    \frac{e^2}{2\pi^2} 
    \!\!\!\!\!\!\!\!\!\int\limits_{[\R-i\delta]\times\R^3} \!\!\!\!\!\!\!\!\! d^4k \, F_\mu(k) A_\nu(-k) 
        (k^\mu k^\nu-g^{\mu\nu}k^2)\!\int\limits_0^1 \! dz\,(z{-}z^2)
        \left[
            \log\left(1-(z{-}z^2)\frac{k^2}{m^2}\right)
            + C 
        \right],
    \label{eq_current_family}
\end{align}
because it fulfils the physical relevant conditions
\ref{Causality}-\ref{Reference current} in second order of perturbation:
\begin{enumerate}
    \item[\ref{Relativistic invariance}]\emph{Relativistic invariance}: The Lorentz covariance is apparent from the Lorentz covariance of
        formula \eqref{eq_def_pi_munu}. Translation invariance follows from the fact that the test functions $F,G$
        appear in energy-momentum space only in the form $k\mapsto F_\mu(k)G_\nu(-k)$.
    \item[\ref{Causality}]\emph{Causality}: Temporal causality follows from the discussed temporal causality of $c^{(2)}_+$; see
        \eqref{eq_c2_eq_diffj}. Fully relativistic causality follows from the temporal causality and relativistic
        invariance.
    \item[\ref{Gauge invariance}]\emph{Gauge invariance}: Gauge invariance is apparent from
        formula \eqref{eq_def_pi_munu} noting $k_\mu (k^\mu k^\nu - g^{\mu\nu}k^2)=0$ and the symmetry in $\mu$ and
        $\nu$.
      \item[\ref{Reference current}]\emph{Reference current}: 
        The vanishing reference current $j_{A=0}=0$
        is already build into the ansatz, see
        \eqref{eq_j_series}, because there is no summand corresponding to $n=1$.
\end{enumerate}
The family of currents~\eqref{eq_current_family} describes the second order term of the vacuum expectation of the
electric current in a prescribed external four-vector potential $G=A$. Concerning the interpretation of the remaining
constant $C$, it is helpful to regard the external current $j_{\text{ext}}^\mu(k) = (k^\mu
k^\nu-g^{\mu\nu}k^2)A^{\text{ext}}_\nu(k)$, associated with the external field $G=A^{\text{ext}}$ by Maxwell's
equations. Changing the integration constant $C$, thus, changes the vacuum polarisation current $j^{(2)}$ in second order of
perturbation by an additional current proportional to $j_{\text{ext}}$. In a self-consistently coupled theory, this
mechanism can be interpreted as to leave the bare electric charge undefined.\\

Be that as it may. Finally and most importantly, we would like to express our gratitude to Detlef for being our
teacher, dear colleague and friend. We miss him dearly.

\appendix

\section{Tetrahedron rule}

Let $A,B,C,D$ be unitary operators on $\cH$ fulfilling~\eqref{eq_ures_invert}
for all $U=X^\inv Y$, $X,Y\in\{A, B, C, D\}$. For such operators we define:
$\Gamma_{ABC}:=\arg\det (A^\inv B)_\Pmm (B^\inv C)_\Pmm (C^\inv A)_\Pmm\in U(1)$

\begin{lemma}\label{lem_Gamma}
    For such operators $A,B,C,D$ we have:
    \textbf{1.} $\Gamma_{ABC}$ is well-defined;
    \textbf{2.} $\Gamma_{ABC}^\inv = \Gamma_{CBA}$;
    \textbf{3.} $\Gamma_{AAB}=1$;\label{tag_tetra_AA}
    \textbf{4.} \label{tag_tetraeder} $\Gamma_{ABC}=\Gamma_{BCD}\Gamma_{DCA}\Gamma_{ABD}$.    
\end{lemma}

\begin{proof}
    \textbf{1:} We observe that each pair of operators such
    as $(A^\inv B)_\Pmm$ is invertible because of $\|\id_\Pmm - (A^\inv
    B)_\Pmm\|<1$. Hence, the Fredholm determinant is well-defined and non-zero
    because of $(A^\inv B)_\Pmm (B^\inv C)_\Pmm (C^\inv A)_\Pmm
        =
        \left(
        (A^\inv C)_\Pmm - (A^\inv B)_\Pmp (B^\inv C)_\Ppm 
        \right) 
        (C^\inv A)_\Pmm$
        \\
        $\in (A^\inv C)_\Pmm (C^\inv A)_\Pmm + I_1(\cH^-)
        = \id_{\cH^-} - (A^\inv C)_\Pmp (C^\inv A)_\Ppm + I_1(\cH^-) =
        \id_\Pmm + I_1(\cH^-)$.

    \noindent \textbf{2:} This can be seen by noting
    $\Gamma_{ABC}^\inv=\overline{\Gamma_{ABC}}$ and $\overline{\det X}=\det
    X^*$ for $X\in\id+I_1$.

    \noindent \textbf{3:}
        $\Gamma_{AAB}=\arg\det (A^\inv B)_\Pmm (B^\inv A)_\Pmm
        = \arg\det |(A^\inv B)_\Pmm|^2 = 1$.

    \noindent \textbf{4:}
    Applying Lidskii’s theorem of cyclic permutation under the
    Fredholm determinant we find
    \begin{align}
        &\Gamma_{BCD}\Gamma_{DCA}\Gamma_{ABD}
        =
        \arg\det 
        (D^\inv B)_\Pmm
        (B^\inv C)_\Pmm
        (C^\inv D)_\Pmm
        \arg\det 
        (D^\inv C)_\Pmm
        (C^\inv A)_\Pmm
        (A^\inv D)_\Pmm
        \notag
        \\
        &\times
        \arg\det 
        (D^\inv A)_\Pmm
        (A^\inv B)_\Pmm
        (B^\inv D)_\Pmm
        =
        \arg\det 
        |(D^\inv B)_\Pmm|^2
        \;
        (B^\inv C)_\Pmm
        \;
        |(D^\inv C)_\Pmm|^2
        \notag
        \\&\circ
        (C^\inv A)_\Pmm
        \;
        |(D^\inv A)_\Pmm|^2
        \;
        (A^\inv B)_\Pmm
        =
        \arg\det 
        (B^\inv C)_\Pmm
        \;
        (C^\inv A)_\Pmm
        \;
        (A^\inv B)_\Pmm
        = \Gamma_{ABC}
        \notag
    \end{align}
    because the operators under the square modulus, like $(B^\inv D)_\Pmm
    (D^\inv B)_\Pmm=|(D^\inv B)_\Pmm|^2$ are positive definite operators in
    $\id_\Pmm + I_1(\cH^-)$.
\end{proof}

\bibliographystyle{plain}

\end{document}